\documentclass[runningheads]{llncs}
\usepackage[T1]{fontenc}
\usepackage{graphicx}

\usepackage{times}
\usepackage{soul}
\usepackage{url}
\usepackage[hidelinks]{hyperref}
\usepackage[utf8]{inputenc}
\usepackage[small]{caption}
\usepackage{graphicx}
\usepackage{amsmath}
\usepackage{booktabs}
\usepackage{algorithm}
\usepackage[switch]{lineno}
\usepackage{etoolbox}
\usepackage[table]{xcolor}
\usepackage{amsfonts}
\usepackage{amssymb}
\usepackage{enumerate}
\usepackage{physics}
\usepackage[T1]{fontenc}
\usepackage{mathrsfs}
\usepackage{multirow}
\usepackage{mathdots}
\usepackage{graphicx}
\usepackage{mathtools}
\graphicspath{ {./images/} }
\usepackage{subfiles}
\usepackage{pgfplots} 
\usepackage{tikz}
\usepackage{algpseudocode}
\usepackage{algorithm}
\usepackage[font=small, margin=0.2\textwidth]{caption}
\usepackage{pythonhighlight}
\usepackage{cancel}
\usepackage[bottom]{footmisc}
\usepackage{subcaption}
\usepackage{caption}
\usepackage{footnote}

\begin{document}

\title{Obvious Independence of Clones \\ \small Extended Abstract}

\titlerunning{Obvious Independence of Clones}
\author{Ratip Emin Berker\inst{1}
\and
Sílvia Casacuberta\inst{2}
\and
Isaac Robinson\inst{2} 
\and
Christopher Ong\inst{3}
}

\authorrunning{Berker et al.}
\institute{Foundations of Cooperative AI Lab (FOCAL), Carnegie Mellon University \and
University of Oxford
 \and
Harvard University\\
}

\maketitle              

\begin{abstract}
The Independence of Clones (IoC) criterion measures a voting rule's robustness to strategic nomination. Prior literature has established empirically that individuals may still submit costly, distortionary misreports even in strategy-proof (SP) settings, due to failure to recognize the SP property. The intersection of these issues motivates the search for mechanisms that are \textit{Obviously Independent of Clones (OIoC)}: where strategic nomination/exiting of clones obviously has no effect on the outcome. We construct a formal and intuitive definition of a voting rule being OIoC and examine five IoC rules to identify whether they satisfy OIoC. 
\end{abstract}

\section{Introduction}\label{sec:intro}
\textit{How can we prevent similar candidates in an election splitting the vote, leading to a less desirable winner}? One answer is to require the voting rule to satisfy the \textit{Independence of Clones (IoC)} criterion, which ensures that the addition or removal of a candidate with similar policy inclinations to others will not spoil the election \cite{tideman87}. While IoC rules  have been well studied in the computational social choice literature \cite{ElkindFaliszewskiSlinko2010,FreemanBrillConitzer2014,ParkesXia2012,tideman87,tideman:CompletIoC}, it is not clear that the average voter or candidate can be easily convinced that any such rule is in fact IoC, resulting, for example, in a candidate unnecessarily dropping out of the race, either out of fear of hurting their party, or of being blamed by their voters for doing so.
    
    As the benefits of a property are sometimes only accrued when agents believe it is satisfied, we turn to examining the \textit{obviousness} of a property. Li \cite{li17} first defined \textit{Obvious Strategy-Proof} mechanisms, which have since been studied in a variety of contexts \cite{ashlagi2018stable,ferraioli2017obvious,gs21}. As we will see, our notion of obviousness for IoC is inspired from the model of primary elections \cite{primaries}, which occurs within each political party (a practical approximation of a clone set) to decide on a joint candidate for a general election.

\section{Preliminaries}\label{sec:pre}

\noindent\textbf{Model}
 Given a finite set of \textit{voters} $N=[n]$ and \textit{candidates} $A=\{a_i\}_{i\in [m]}$, each $i \in N$ has a strict \textit{ranking} $\sigma_i$ over $A$. A \textit{preference profile} $\boldsymbol{\sigma}$ consists of all voters' rankings. A \textit{voting rule} is a function that maps $\boldsymbol{\sigma}$ to a subset of $A$, the winner(s) of the election.

\begin{definition}[Independence of Clones \cite{tideman87}]\label{def:clone_set} A non-empty subset of candidates, $K \subseteq A$, is a \textbf{set of clones} with respect to $\boldsymbol{\sigma}$ if no voter ranks any candidate outside of $K$ between any two elements of $K$. We say a voting rule is \textit{Independent of Clones (IoC)} if:
\begin{enumerate}
    \item A candidate that is a member of a set of clones wins if and only if some member of that set of clones wins after one of its clones is eliminated from the original ballot.
    \item A candidate that is \textit{not} a member of a set of clones wins if and only if that candidate wins after any clone is eliminated from the original ballot.
    
\end{enumerate}

\end{definition}
Intuitively, a rule satisfying IoC ensures that the winner of an election does not change due to the addition of a candidate who is similar to an existing non-winning candidate. 
\\~\\
\noindent\textbf{Voting Rules Considered}\label{sec:functions} We study five existing IoC rules, the definitions of which we provide in the full paper: \textit{Single Transferable Vote (STV)}, \textit{Ranked Pairs (RP) \cite{tideman87}}, \textit{the Schulze method \cite{schulze10}}, \emph{Schwartz rule \cite{debrajat}}, and \textit{Smith Alternative Vote (SAV) \cite{green2011four}}.

\section{Obvious Independence of Clones (OIoC) and Results}\label{sec:oioc}
  
Before introducing the definition of OIoC, we first introduce two novel concepts:

  \begin{definition}\label{def:clone_grouping}
Given a preference profile $\boldsymbol{\sigma}$, a set of sets $\mathcal{K}=\{K_1,K_2,\ldots,K_\ell\}$ where $K_i \subseteq A$ for all $i\in [\ell]$ is a \textbf{clone partition with respect to} $\boldsymbol{\sigma}$ if (1) $\mathcal{K}$ is a disjoint partitioning of $A$, and (2) each $K_i$ is a non-empty clone set with respect to $\boldsymbol{\sigma}$.
\end{definition}
\begin{definition}[GLOC]\label{def:gloc}
    Given any voting rule $f$, $\boldsymbol{GLOC_f}$ is a function that takes as input a preference profile $\boldsymbol{\sigma}$ and a clone partition $\mathcal{K}$ with respect to $\boldsymbol{\sigma}$ and performs:
    
    \smallskip
        \noindent \textit{1. \textbf{GLO}bal step:} Given $\sigma_i$, say $\sigma^\mathcal{K}_i$ is the corresponding ranking over $\mathcal{K}$. Treating $\mathcal{K}$ as the set of candidates, compute $f\left(\boldsymbol{\sigma}^\mathcal{K}\right) \equiv f\left(\{{\sigma}_i^\mathcal{K}\}_{i \in N}\right)$ to get the `winner clone sets'.
        
             \smallskip
        \noindent \textit{2. \textbf{LOC}al step:} For each $K \in f\left(\boldsymbol{\sigma}^\mathcal{K}\right)$, say $\boldsymbol{\sigma}^{K}$ is $\boldsymbol{\sigma}$ with $A\setminus K$ removed. Compute $f(\boldsymbol{\sigma}^K)$ and output the union over all $K$ such, i.e., $GLOC_f(\boldsymbol{\sigma}, \mathcal{K}) = \bigcup_{K \in f\left(\boldsymbol{\sigma}^\mathcal{K}\right)  }f(\boldsymbol{\sigma}^K)$
\end{definition}

Intuitively, GLOC first runs the input voting rule $f$ on the clone sets (as specified by $\mathcal{K}$), `packing' the candidates in each set to treat it as a candidate. It then `unpacks' the clones within each winner clone set, and runs $f$ once again among them. To demonstrate GLOC, we illustrate the protocol when applied to the setting of plurality voting $f_{plur}$, which simply picks the candidate who is the top choice for the most voters. 

\begin{figure}[h!]
\centering
 \begin{tabular}{|c|c|c|c|c|c|c|}
  \hline
 Vot.1&Vot.2&Vot.3&Vot.4&Vot.5&Vot.6&Vot.7\\ \hline 
  \cellcolor{yellow!25}{a} & \cellcolor{yellow!25}{a} &\cellcolor{blue!25}{b} & \cellcolor{blue!25}{b} & \cellcolor{green!25}{c} & \cellcolor{green!25}{c} & \cellcolor{green!25}{c}\\
  \hline
   \cellcolor{blue!25}{b} &  \cellcolor{blue!25}{b} &\cellcolor{yellow!25}{a} & \cellcolor{yellow!25}{a} & \cellcolor{red!25}{d} & \cellcolor{red!25}{d} & \cellcolor{red!25}{d} \\
  \hline
     \cellcolor{green!25}{c} & \cellcolor{red!25}{d} & \cellcolor{green!25}{c} & \cellcolor{red!25}{d} &  \cellcolor{yellow!25}{a} & \cellcolor{yellow!25}{a} & \cellcolor{blue!25}{b} \\
  \hline
     \cellcolor{red!25}{d} & \cellcolor{green!25}{c}  & \cellcolor{red!25}{d} & \cellcolor{green!25}{c} & \cellcolor{blue!25}{b} & \cellcolor{blue!25}{b} & \cellcolor{yellow!25}{a} \\
    \hline
\end{tabular}

    \captionsetup{width=.8\linewidth}
    \caption{An example preference profile $\boldsymbol{\sigma}$. Column Vot.$i$ shows $\sigma_i$ for Voter $i$. }
    \label{fig:GLOC_Example_1}
\end{figure}
Consider $\boldsymbol{\sigma}$ from Figure \ref{fig:GLOC_Example_1}.  We have $f_{plur}(\boldsymbol{\sigma})=c$. Notice $\mathcal{K}=\{\{a,b\}, \{c,d\}\}$ is a valid clone partition with respect to $\boldsymbol{\sigma}$. Accordingly, GLOC maps $\{a,b\}$ and $\{c,d\}$ to the meta-candidates $K_1$ and $K_2$, respectively. As demonstrated in Figure \ref{fig:GLOC_Example_2}, we have $f_{plur}(\boldsymbol{\sigma}^\mathcal{K})=K_1$. and  $f_{plur}(\boldsymbol{\sigma}^{K_1})=a$, implying $GLOC_{f_{plur}}(\boldsymbol{\sigma}, \mathcal{K}) =a$.

\begin{figure}[h!]
\centering

 \begin{tabular}{|c|c|c|c|c|c|c|}
  \hline Vot.1&Vot.2&Vot.3&Vot.4&Vot.5&Vot.6&Vot.7\\ \hline 
  \cellcolor{cyan!25}{$K_1$} & \cellcolor{cyan!25}{$K_1$} &\cellcolor{cyan!25}{$K_1$} & \cellcolor{cyan!25}{$K_1$} & \cellcolor{brown!25}{$K_2$} & \cellcolor{brown!25}{$K_2$} & \cellcolor{brown!25}{$K_2$}\\
  \hline
     \cellcolor{brown!25}{$K_2$} & \cellcolor{brown!25}{$K_2$} & \cellcolor{brown!25}{$K_2$} & \cellcolor{brown!25}{$K_2$} &  \cellcolor{cyan!25}{$K_1$} & \cellcolor{cyan!25}{$K_1$} & \cellcolor{cyan!25}{$K_1$} \\
    \hline
\end{tabular}
\quad \quad
 \begin{tabular}{|c|c|c|c|c|c|c|}
  \hline
 
  Vot.1&Vot.2&Vot.3&Vot.4&Vot.5&Vot.6&Vot.7\\ \hline 
  \cellcolor{yellow!25} a & \cellcolor{yellow!25}  a &\cellcolor{blue!25} b & \cellcolor{blue!25} b &  \cellcolor{yellow!25} a & \cellcolor{yellow!25} a & \cellcolor{blue!25} b\\
  \hline
   \cellcolor{blue!25} b &  \cellcolor{blue!25} b &\cellcolor{yellow!25} a & \cellcolor{yellow!25} a & \cellcolor{blue!25} b & \cellcolor{blue!25} b & \cellcolor{yellow!25} a \\
  \hline
\end{tabular}

    \captionsetup{width=.8\linewidth}
    \caption{(Left) $\boldsymbol{\sigma}^\mathcal{K}$, where the clone sets are condensed into singular candidates $K_1$ and $K_2$. (Right) $\boldsymbol{\sigma}^{K_1}$, where each $\sigma_i$ is limited to the  members of $K_1$.}
    \label{fig:GLOC_Example_2}
\end{figure}

Having defined Clone Partitions and GLOC, we now formally introduce OIoC.
  \begin{definition}\label{def:oioc}
      A voting rule $f$ is \textbf{Obviously Independent of Clones (OIoC)} if for all preference profile $\boldsymbol{\sigma}$ and all clone partitions $\mathcal{K}$ w.r.t. $\boldsymbol{\sigma}$, we have $f(\boldsymbol{\sigma}) =GLOC_f(\boldsymbol{\sigma}, \mathcal{K})$
\end{definition}

The example from Figures \ref{fig:GLOC_Example_1} and \ref{fig:GLOC_Example_2} demonstrate that plurality is not OIoC, which is not surprising, considering the rule is not even IoC (as having a clone will split your plurality votes). In the full version of the paper, we formalize this hierarchy:
\begin{proposition}
    OIoC implies IoC.
\end{proposition}

Most real-life elections do not result in ties. If we instead require agreement between $f$ and $GLOC_f$ only when there is a clear winner, we get a natural relaxation of OIoC:
  \begin{definition}\label{def:woioc}
      A voting rule $f$ is \textbf{weakly Obviously Independent of Clones (wOIoC)}  if given any  $\boldsymbol{\sigma}$ and any clone partition $\mathcal{K}$ w.r.t. $\boldsymbol{\sigma}$,  $f(\boldsymbol{\sigma}) = \{a\}$ iff $GLOC_f(\boldsymbol{\sigma}, \mathcal{K})=\{a\}$
\end{definition}

OIoC clearly implies wOIoC. The relationship between wOIoC and IoC is more nuanced: they are incomparable; however, (as discussed in the full version of the paper), wOIoC implies IoC under some reasonable assumptions about the voting rule. Having established the definitions of OIoC and wOIoC, we prove which rule satisfies which:
\begin{theorem}
    STV, the Schulze method, and SAV are not (even weakly) OIoC. Schwartz rule is wOIoC. Ranked pairs is OIoC.
\end{theorem}
Of the five results, the most sophisticated is  RP being OIoC. The proof depends on \textit{impartial tie-breaking}, defined in \cite{tideman:CompletIoC}, which is required for RP to always satisfy IoC.

\section{Conclusion \& Future Work} \label{conc}

Definition \ref{def:oioc} has a natural interpretation: if a voting rule is OIoC, then the outcome of the election will be same regardless of whether we (1) apply the rule directly or we (2) let the parties (clone sets) run primaries (pick their `best' member) and run the election among these winners. Apart from this consistency result, it also has practical implications: if a rule is OIoC, the decision of a candidate to opt-out of an election can be postponed until after the winning parties are computed, hence removing any concern over a candidacy resulting in the loss of their party. Additionally, OIoC allows cutting expenses by eliminating primaries. Without primaries, OIoC rules can also derive clone sets a posteriori from the votes, rather than assuming a political party to be a clone set.

This paper opens several new lines of work. For instance, one could study the problem of extending other axioms from social choice (such as monotonicity or independence of irrelevant alternatives) to fit the framework of obviousness.  More broadly, studying obviousness not only from a computational perspective but also an empirical or psychological point of view may shed light on how best to approach defining the obviousness of other axiomatic properties.

 \bibliographystyle{splncs04}
 \bibliography{refs}

\end{document}